\documentclass{degruyter-proceedings}          
\usepackage{dsfont}
\usepackage[colorlinks=true,linkcolor=blue,citecolor=red]{hyperref}

\theoremstyle{definition}
\newtheorem{notation}[theorem]{Notation}

\evensidemargin=\oddsidemargin

\newcommand{\grad}{\nabla}
\newcommand{\grap}{\nabla_\perp}

\newcommand{\Div}{\operatorname{\mathrm{div}}}
\newcommand{\Divp}{\operatorname{\mathrm{div}}_\perp}
\newcommand{\rot}{\operatorname{\mathrm{curl}}}
\newcommand{\rots}{\operatorname{\textit{curl}}_\perp}
\newcommand{\rotv}{\operatorname{\textup{\textbf{curl}}}_\perp}
\newcommand{\rotw}{\operatorname{\widetilde{\textup{\textbf{curl}}}}_\perp }
\newcommand{\Deltap}{\Delta_\perp}

\newcommand  {\N}{{\mathbb N}}

\newcommand  {\R}{{\mathbb R}}
\newcommand  {\T}{{\mathbb T}}
\newcommand  {\Z}{{\mathbb Z}}

\renewcommand{\SS}{\mathbb S}

\newcommand  {\ba}{{\boldsymbol{a}}}
\newcommand  {\bx}{{\boldsymbol{x}}}
\newcommand  {\by}{{\boldsymbol{y}}}
\newcommand  {\bz}{{\boldsymbol{z}}}

\newcommand  {\cY}{\mathcal{Y}}
\newcommand  {\cZ}{\mathcal{Z}}

\newcommand  {\gH}{\mathfrak{H}}
\newcommand  {\gL}{\mathfrak{L}}

\newcommand  {\rd}{\mathrm{d}}

\newcommand  {\EE}{\boldsymbol{\mathsf E}}

\newcommand  {\HH}{\boldsymbol{\mathsf H}}
\newcommand  {\MM}{{\boldsymbol{\mathsf M}}}
\newcommand  {\NN}{{\boldsymbol{\mathsf N}}}

\newcommand  {\XX}{\boldsymbol{\mathsf X}}

\newcommand  {\cc}{{\boldsymbol{\mathsf c}}}
\newcommand  {\ee}{{\boldsymbol{\mathsf e}}}
\newcommand  {\nn}{\boldsymbol{\mathsf n}}

\newcommand  {\uu}{\boldsymbol{\mathsf u}}
\newcommand  {\vv}{\boldsymbol{\mathsf v}}

\newcommand  {\xx}{\boldsymbol{\mathsf x}}

\newcommand  {\mm}{\mathsf n_{\mathsf{opt}}}

\newcommand{\iti}[1]{\noindent\mbox{{\em (\romannumeral #1)\/}}}
\newcommand{\db}[1]{_{\raise-0.3ex\hbox{$\scriptstyle #1$}}}
\newcommand{\dd}[1]{_{\raise-1.5pt\hbox{$\scriptstyle #1$}}}

\newcommand{\di}{\displaystyle}

\newcommand{\dr}{{\rm d}}
\newcommand{\el}{_{\mathsf N}}
\newcommand{\ma}{_{\mathsf T}}
\newcommand{\xp}{x_\perp}
\newcommand{\up}{\uu_\perp}
\newcommand{\vp}{\vv_\perp}

\newcommand{\vvkj}{\vv^m_{j}}
\newcommand{\vpkj}{\vv^m_{\perp,j}}
\newcommand{\np}{\nn_\perp}
\newcommand{\neu}{^{\sf neu}}
\newcommand{\dir}{^{\sf dir}}
\newcommand{\mix}{^{\sf mix}}
\newcommand{\too}{^{\sf top}}
\newcommand{\TE}{^{\sf TE}}
\newcommand{\TM}{^{\sf TM}}
\newcommand{\TEM}{^{\sf TEM}}
\newcommand{\rel}{_{\sf rel}}
\newcommand{\cnd}{_{\sf cd}}
\newcommand{\ins}{_{\sf ins}}

\title{Maxwell eigenmodes in product domains}

\lastnameone{Costabel}
\firstnameone{Martin}
\nameshortone{M.~Costabel}
\addressone{Univ. Rennes, CNRS, IRMAR - UMR 6625, F-35000 Rennes}
\countryone{France}
\emailone{Martin.Costabel@univ-rennes1.fr}

\lastnametwo{Dauge}
\firstnametwo{Monique}
\nameshorttwo{M.~Dauge}
\addresstwo{Univ. Rennes, CNRS, IRMAR - UMR 6625, F-35000 Rennes}
\countrytwo{France}
\emailtwo{Monique.Dauge@univ-rennes1.fr}

\abstract{This paper is devoted to Maxwell modes in three-dimensional bounded electromagnetic cavities that have the form of a product of lower dimensional domains in some system of coordinates. The boundary conditions are those of the perfectly conducting or perfectly insulating body. The main case of interest is products in Cartesian variables. Cylindrical and spherical variables are also addressed. We exhibit common structures of polarization type for eigenmodes. In the Cartesian case, the cavity eigenvalues can be obtained as sums of Dirichlet or Neumann eigenvalues of positive Laplace operators and the corresponding eigenvectors have a tensor product form. We compare these descriptions with the spherical wave function Ansatz for a ball and show why the cavity eigenvalue of the ball are also Dirichlet or Neumann eigenvalues of some scalar operators. As application of our general formulas, we find explicit eigenpairs in a cuboid, in a circular cylinder, and in a cylinder with a coaxial circular hole. This latter example exhibit interesting ``TEM'' eigenmodes that have a one-dimensional vibrating string structure, and contribute to the least energy modes if the cylinder is long enough.
}



\researchsupported{The authors belong to the Centre Henri Lebesgue ANR-11-LABX-0020-01.}

\firstpage{1}

\begin{document}


\section{Introduction}
\label{S1}
A domain $\Omega$ of $\R^n$ is called a product domain if for a choice of Cartesian coordinates $\bx=(\by,\bz)$ in $\R^n$, the domain $\Omega$ coincides with the product $\cY\times\cZ$ in the sense that
\[
   \bx\in\Omega\quad\Longleftrightarrow\quad\by\in\cY \ \ \mbox{and} \ \ \bz\in\cZ.
\]
In the three-dimensional space ($n=3$), we may assume without restriction that $\cY$ has the dimension $2$, and $\cZ$, dimension $1$, hence is an interval. Such a domain may also be called a cylinder.
The main motivation of this work is to exhibit for electromagnetic cavity problems in a cylinder with arbitrary cross section similar properties as those, well known, for acoustic modal problems.

Most of the results we present are not new and have their roots in the pioneering works by Mie (1908) and Debye (1909). Expressions for cavity modes in cylinders can be found in \cite{Jackson1998} and in balls in \cite{HansonYakovlev}. Our aim is to adopt a synthetic presentation that clearly links Laplace or Laplace-like eigenvectors to electromagnetic eigenmodes via TE (transverse electric) and TM (transverse magnetic) vector wave functions: The Laplace eigenvectors appear as Debye potentials. In particular we carefully address the case when the cross section $\omega$ of $\Omega$ contains holes (modelling for instance metallic wires) and prove the completeness of a system of TE, TM and TEM modes. The TEM eigenmodes that enjoy both features of transverse electric and magnetic polarizations, often contribute the lowest frequencies, and this can be precisely quantified. This case was the first motivation for the present investigation.

The knowledge of Maxwell eigenmodes to an applied mathematics audience has some importance. Our results can be used as benchmarks for numerical methods for the computation of cavity modes. 
Also for transmission problems, our description of the interior eigenmodes may be useful, since there exist standard numerical methods that fail if the frequency coincides with an interior eigenfrequency.

\subsection{The case of acoustics: The Dirichlet Laplacian}
The Laplace operator $\Delta$ in $\R^n$ is expressed in variables $\bx=(x_1,\ldots,x_n)$ as $\Delta=\sum_{1\le j\le n}\partial^2_{x_j}$ and it is the sum of the two Laplace operators in variables $\by$ and $\bz$
\[
   \Delta = \Delta_\by + \Delta_\bz\,.
\]
The Sobolev space $H^1(\Omega)$ on the product domain $\Omega=\cY\times\cZ$ can be written as
\[
   H^1(\Omega) = L^2(\cY,H^1(\cZ)) \cap H^1(\cY,L^2(\cZ)).
\]
Likewise, the closure $H^1_0(\Omega)$ in $H^1(\Omega)$ of smooth functions with compact support in $\Omega$ satisfies 
\[
   H^1_0(\Omega) = L^2(\cY,H^1_0(\cZ)) \cap H^1_0(\cY,L^2(\cZ)).
\]
As a direct consequence we find, for any bounded product domain $\Omega$, the full spectral description of the Dirichlet Laplacian:

\begin{theorem}
\label{th:1}
Let $(\lambda_j,v_j)_{j\ge1}$ and $(\mu_m,w_m)_{m\ge1}$ be the spectral sequences of $-\Delta_\by$ on $H^1_0(\cY)$ and of $-\Delta_\bz$ on $H^1_0(\cZ)$, respectively. This means that
\[
   \lambda_1<\lambda_2\le\ldots
\]
is the eigenvalue sequence of $-\Delta_\by$ and $(v_j)_{j\ge1}$ is an associated orthonormal basis, and the same for $-\Delta_\bz$. 

Then the set of eigenvalues of $-\Delta$ on $H^1_0(\cY\times\cZ)$ is
\[
   \{\lambda_j+\mu_m,\quad j\ge1,\ \ m\ge1\}
\]
and the tensor functions $u_{j,m} := v_j\otimes w_m$, i.e. defined as
\[
   u_{j,m}(\bx) = v_j(\by)w_m(\bz),
\]
are orthonormal associated eigenvectors, and they form a basis of $H^1_0(\Omega)$.
\end{theorem}

Of course, a similar result holds for Neumann boundary conditions. This holds also for mixed Dirichlet-Neumann problems of the type
\[
   \mbox{(Dirichlet on $\partial\cY\times\cZ$) \quad and \quad (Neumann on $\cY\times\partial\cZ$)}
\]
for which the $(\lambda_j,v_j)$ are still the Dirichlet eigenpairs on $\cY$, but the $(\mu_m,w_m)$ have to be taken as the Neumann eigenpairs on $\cZ$. Finally the Dirichlet and Neumann conditions can be also be swapped between $\cY$ and $\cZ$.

\subsection{The case of electromagnetism: The Maxwell system}
From now on the space dimension is $n=3$. 
Let $\Omega$ be a domain in $\R^3$, representing a cavity filled by an homogeneous dielectric medium. 
We assume that the boundary of $\Omega$ represents {\em perfectly conducting} walls. After normalization, the cavity resonator problem is to find the frequencies $k\in\R$ and the non-zero electromagnetic fields $(\EE,\HH)$ in $L^2(\Omega)^6$ such that
\begin{equation}
\label{0E2}
   \left\{
   \begin{array}{ll}
   \rot\EE - i k\HH = 0 & \mbox{in}\quad\Omega,\\
   \rot\HH + i k\EE = 0 & \mbox{in}\quad\Omega,\\
   \Div\EE=0 \quad\mbox{and}\quad \Div\HH=0\quad & \mbox{in}\quad\Omega,\\
   \EE\times\nn = 0 \quad\mbox{and}\quad\HH\cdot\nn=0, 
   & \mbox{on}\quad\partial\Omega.
   \end{array}
   \right.
\end{equation}
Here, $\nn$ denotes the outward unit normal to $\partial\Omega$. The gauge conditions on the divergence are a consequence of the first two equations if $k\neq0$. Nevertheless we look for solutions of \eqref{0E2} including $k=0$. The occurrence of $k=0$ happens if and only if the domain $\Omega$ is topologically non-trivial, i.e.\ if $\Omega$ is not simply connected, or if $\partial\Omega$ is not connected, see Propositions 3.14 \& 3.18 in reference \cite{AmroucheBernardiDaugeGirault98}.

\begin{definition}
\label{0D1}
The triples $(k, \EE, \HH)$ solution of \eqref{0E2} with $(\EE, \HH)\neq0$ are called Maxwell eigenmodes, $k$ is called eigenfrequency, $k^2$ eigenvalue, and $\EE$, $\HH$ electric and magnetic eigenvectors.
\end{definition}

Let $\Omega$ be a bounded product domain in $\R^3$. This means that 
\begin{equation}
\label{1E2}
   \Omega = \omega\times I,\quad \omega\subset\R^2,\ \ I \ \mbox{interval in}\ \R.
\end{equation}
We denote correspondingly Cartesian coordinates in $\Omega$ by
$$
   x = (x_1,x_2,x_3) = (\xp,x_3),\quad \xp\in\omega \ \ \mbox{and} \ \ x_3\in I.
$$
We assume that $\omega$ is a bounded Lipschitz domain.  We note that the boundary of $\Omega$ is connected. But, if $\omega$ is not simply connected, the same holds for $\Omega$.

\begin{notation}
\label{not:1}
\begin{enumerate}
\item Denote by $\Deltap=\partial^2_1+\partial^2_2$ the Laplace operator in the variables $\xp$.
\item Let $\big(\lambda\dir_j,v\dir_j\big)_{j\ge1}$ be the eigenpair sequence of the Dirichlet problem in $\omega$ for the operator $-\Deltap$.
\item Let $\big(\lambda\neu_j,v\neu_j\big)_{j\ge0}$ be the eigenpair sequence of the Neumann problem in $\omega$ for the operator $-\Deltap$, with $\lambda\neu_0=0$ and $v\neu_0=1$.
\item  Let $\big(\mu\dir_m,w\dir_m\big)_{m\ge1}$ be the eigenpair sequence of the Dirichlet problem in $I$ for the operator $-\partial^2_3$.
\item Let $\big(\mu\neu_m,w\neu_m\big)_{m\ge0}$ be the eigenpair sequence of the Neumann problem in $I$ for the operator $-\partial^2_3$, with $\mu\neu_0=0$ and $w\neu_0=1$.
\end{enumerate}
\end{notation}

One of the results of this paper is (see Theorem \ref{2T1})

\begin{theorem}
\label{th:2}
Assume that $\omega$ is simply connected. Then the Maxwell eigenvalues $k^2$ span the set
\begin{equation}
\label{eq:th2}
   \big\{\lambda\dir_j+\mu\neu_m,\quad j\ge1,\ m\ge0\big\} \ \cup\ 
   \big\{\lambda\neu_j+\mu\dir_m,\quad j\ge1,\ m\ge1\big\}.
\end{equation}
(including repetition according to multiplicities).
\end{theorem}

In the sequel we describe a corresponding basis of eigenvectors, constructed on the model of vector wave functions, according to the widely used $\MM$ and $\NN$ ansatz (Debye potentials). We include the case when $\omega$ is multiply connected: In this case, the relevant parameter is the number $D$ of connected components of $\partial\omega$ and to the set \eqref{eq:th2}, we have to add all the $\mu\dir_m$, each of them with multiplicity $D-1$, corresponding to the number of holes contained in $\omega$. The corresponding modes are the TEM modes that have no component in the direction $x_3$.

This paper is organized as follows. In section \ref{S2} we introduce general principles for the description of the Maxwell cavity modes. In section \ref{S3} we give formulas for the electric eigenmodes $(\kappa^2,\EE)$ in the case when $\Omega$ has the cylindric form $\omega\times I$ with $\omega\subset\R^2$ and $I\subset\R$, separating the modes according to their polarization in TE, TM and TEM types. In section \ref{S3M} we deduce the structure of magnetic cavity modes and synthesize results in Table \ref{tab:1}. 
In section \ref{S3x} we mention generalizations to special combinations of conducting and insulating boundary conditions.

As an application of our formulas, we consider in section \ref{S4} the case when $\Omega$ is a cube (or, more generally, a cuboid), and in section \ref{S5} the case when $\Omega$ is axisymmetric: Then $\Omega$ is a circular cylinder, or a circular cylinder with a coaxial cylindrical hole. We bring special attention to the latter case. 
Then the TEM modes appear in the explicit form \eqref{5E15}.

We address the situation when $\Omega$ is a ball of radius $R$ in section \ref{sec:sph}. The analysis is in the same spirit and exhibits a close relation with a scalar Laplace-like operator in the ``cylinder'' $\SS^2\times(0,R)$.

Finally, in section \ref{S7}, again for product domains, we investigate the variable coefficient case, namely when $\varepsilon$ is varying transversally, i.e.\  independently of the axial variable $x_3$. Then the TE and TM structures are no longer a valid Ansatz, in general. In replacement, we obtain wave guide formulations with separation of variables and tensor product form for eigenmodes.

\section{Preliminaries}
\label{S2}

\subsection{Electric and magnetic formulations for the Maxwell spectrum}
\label{se:Max}
We first recall the definition of the standard functional spaces associated with
Maxwell equations on a domain $\Omega\subset\R^3$. The curl in 3D is defined as
\[
   \rot\uu = \begin{pmatrix}
      \partial_2 u_3 - \partial_3 u_2 \\ 
      \partial_3 u_1 - \partial_1 u_3 \\ 
      \partial_1 u_2 - \partial_2 u_1 \\ 
      \end{pmatrix}
    \quad\mbox{for}\quad \uu=(u_1,u_2,u_3)
\]
and $\HH(\rot,\Omega)$ is the space of $L^2(\Omega)$ fields
with curl in $L^2(\Omega)$, while $\HH_0(\rot,\Omega)$ is the subspace of $\HH(\rot,\Omega)$
with perfectly conducting electric boundary condition $\uu\times\nn=0$.

The divergence in 3D is defined as
\[
   \Div\uu = \partial_1u_1+\partial_2u_2+\partial_3u_3
    \quad\mbox{for}\quad \uu=(u_1,u_2,u_3)
\]
and $\HH(\Div,\Omega)$ is the space of $L^2(\Omega)$ fields
with divergence in $L^2(\Omega)$, and $\HH_0(\Div,\Omega)$ the subspace of $\HH(\Div,\Omega)$
with perfectly conducting magnetic boundary conditions $\uu\cdot\nn=0$. 

It is well known that the system of equations \eqref{0E2} can be formulated with $\EE$ only (electric formulation) or $\HH$ only (magnetic formulation). Each time a vector Helmholtz equation is found. Convenient functional spaces for the electric and magnetic variational formulations are 
$$
\XX\el(\Omega) :=  \HH_0(\rot,\Omega) \cap \HH(\Div,\Omega)
\quad\mbox{and}\quad
\XX\ma(\Omega) :=  \HH(\rot,\Omega) \cap \HH_0(\Div,\Omega) \: .
$$
In these spaces, \emph{regularized} formulations make sense. This means that, introducing a parameter 
\[
   s\ge0
\]
we introduce the electric variational formulations:\\[1ex] 
{\it Find the eigenpairs $(\Lambda=\kappa^2,\uu)$ with $\uu\ne0$ in $\XX\el(\Omega)$ such that}
\begin{equation}
    \int_\Omega \rot\uu\,\rot\vv + s\Div\uu\Div\vv\;\dr\xx = \Lambda \int_\Omega
    \uu\cdot\vv\;\dr\xx, \quad \forall\vv\in\XX\el(\Omega),
\label{1E1}
\end{equation}
while magnetic formulations are:\\[1ex] 
{\it Find the eigenpairs $(\Lambda=\kappa^2,\uu)$ with $\uu\ne0$ in $\XX\ma(\Omega)$  such that}
\begin{equation}
    \int_\Omega \rot\uu\,\rot\vv + s\Div\uu\Div\vv\;\dr\xx = \Lambda \int_\Omega
   \uu\cdot\vv\;\dr\xx,
   \quad \forall\vv\in\XX\ma(\Omega).
\label{1E1m}
\end{equation}

Relying on \cite[Theorem 1.1]{CoDa1999M2AS}, we know that
the eigenpairs of \eqref{1E1} split in two families
\begin{itemize}
\item[a)] the Maxwell eigenvalues, independent of $s$, for which the eigenvectors are divergence free,
\item[b)] the gradients of the Dirichlet eigenvectors for $-\Delta$ on $\Omega$, associated with eigenvalues $s\lambda\dir_\Omega$.
\end{itemize}
Thus the regularization by $s\Div\uu\Div\vv$ makes the problem elliptic as soon as $s>0$ and gives a description of the infinite dimensional kernel of the $\rot\rot$ operator. The gauge conditions $\Div\EE=0$ and $\Div\HH=0$ in \eqref{0E2} ensure that we are always in case a). We can state

\begin{lemma}
\label{1L1}
\begin{enumerate}
\item Let $(k,\EE,\HH)$ be a Maxwell eigenmode solution of \eqref{0E2}. Set $\Lambda =k^2$. Then, if $\EE\neq0$, it is solution of \eqref{1E1} for any $s\ge0$, and if $\HH\neq0$, it is solution of \eqref{1E1m} for any $s\ge0$.
\item Let $s\ge0$. If $\Lambda\neq 0$ and $\uu$ is solution of \eqref{1E1} with $\Div\uu=0$, then setting $k=\pm\sqrt{\Lambda}$, $\EE=\uu$, and $\HH = \frac1{ik}\rot\EE$, we obtain an eigenmode of \eqref{0E2}.
\item Let $s\ge0$. If $\Lambda\neq 0$ and $\uu$ is solution of \eqref{1E1m} with $\Div\uu=0$, then setting $k=\pm\sqrt{\Lambda}$, $\HH=\uu$, and $\EE = -\frac1{ik}\rot\HH$, we obtain an eigenmode of \eqref{0E2}.
\end{enumerate}
\end{lemma}

\subsection{Product domain}

Let $\Omega\subset\R^3$ be of product form $\omega\times I$, with $\omega\subset\R^2$ and an interval $I$.
We denote Cartesian coordinates and component of vectors as
\[
   x = (x_1,x_2,x_3) = (\xp,x_3)\quad\mbox{and}\quad \uu = (u_1,u_2,u_3) = (\up,u_3).
\]
Likewise, the exterior unit normal $\nn$ to $\partial\Omega$ is written $(\np,n_3)$. The boundary of $\Omega$ is
\[
   \partial\Omega = (\partial\omega\times \overline I) \ \cup \ (\overline\omega \times \partial I).
\]
On $\omega\times\partial I$, $\np=0$ and $n_3=\pm1$. On $\partial\omega\times I$, $\np$ is the exterior unit normal to $\partial\omega$, $n_3=0$, and the tangential component of $\up$ is $\up\times\np=u_1n_2-u_2n_1$.
The {\em electric boundary conditions $\uu\times\nn=0$ on $\partial\Omega$} are equivalent to
\begin{equation}
\label{1E4}
\begin{split}
   &\up\times\np = 0 \quad\mbox{and}\quad u_3=0 \quad\mbox{on}\quad\partial\omega\times I,
   \\
   &\up=0  \quad\mbox{on}\quad\omega\times\partial I,
\end{split}
\end{equation}

The gradient and the Laplacian in the transverse plane containing $\omega$ are denoted by $\grap$ and $\Deltap$:
\[
   \grap v = \begin{pmatrix}
      \partial_1 v \\ 
      \partial_2 v \\ 
      \end{pmatrix} \quad\mbox{and}\quad 
    \Deltap v =  \partial^2_1 v + \partial^2_2 v.
\]
The vector and scalar curls in 2D are given by:
\[
   \rotv v = \begin{pmatrix}
      \partial_2 v \\ 
       - \partial_1 v \\ 
      \end{pmatrix} \quad\mbox{and}\quad 
    \rots \vv =  \partial_1 v_2 - \partial_2 v_1.
\]
We have the formula
\begin{equation}
\label{1E3}
   \rot\uu = \begin{pmatrix}
       \rotv u_3 \\ \rots \up
       \end{pmatrix} +
       \partial_3 \begin{pmatrix}
       -u_2 \\ u_1 \\ 0
       \end{pmatrix}.
\end{equation}

\subsection{The $\MM$, $\NN$ ansatz and the TE or TM polarizations}

The interior partial differential equation satisfied by eigenpairs is the system:
\begin{equation}
\label{1E1b}
   \rot\rot\uu = k^2\,\uu \quad\mbox{and}\quad \Div\uu=0\quad\mbox{in}\quad\Omega.
\end{equation}
There is a well known ansatz to solve these equations, called vector wave functions $\MM$ and $\NN$. They depend on the choice of a unit \emph{piloting vector} $\hat\cc$, and then $\MM$ and $\NN$ are generated by scalar potentials $q=q(x)$ according to
\begin{equation}
\label{eq:MN}
   \MM[q] = \rot(q\,\hat\cc) \quad\mbox{and}\quad \NN[q] = \rot\MM[q] = \rot\rot(q\,\hat\cc).
\end{equation}
In a slightly modified form where one takes $\hat\cc=\frac{x}{|x|}$, the ansatz $\MM$ and $\NN$ are the corner stone for the construction of spherical wave functions, cf.\ section \ref{sec:sph}.

For our study, we choose
\begin{equation}
\label{eq:pilot}
   \hat\cc = \ee_3 = \begin{pmatrix}
   {\:0\:} \\ 0\\ 1
\end{pmatrix}.
\end{equation}
Direct calculations yield:
\begin{lemma}
Let $q\in H^1(\Omega)$ and set $\MM[q] = \rot(q\,\ee_3)$. Then
\begin{equation}
\label{eq:Mrot}
   \MM[q] = \begin{pmatrix}
         \rotv q \\ 0
         \end{pmatrix}
   \quad\mbox{and}\quad
   \rot\MM[q] = \nabla(\partial_3q) - \Delta q\;\ee_3.
\end{equation}
With $\NN[q]=\rot\MM[q]$, we have
\begin{equation}
\label{eq:Nrot}
   \NN[q] = \nabla(\partial_3q) - \Delta q\;\ee_3
   \quad\mbox{and}\quad
   \rot\NN[q] = -\MM[\Delta q].
\end{equation}
\end{lemma}
The form of $\MM$ and $\rot\NN$ with their third component zero explains why $\MM$, when describing an electric field, represents the TE (transverse electric) polarization, and $\NN$, the TM (transverse magnetic) polarization. For the description of a magnetic field, the converse happens: $\MM$ is TM and $\NN$ is TE.

As a consequence, we find that
\begin{equation}
\label{eq:MNH}
\begin{gathered}
   (\rot\rot-k^2)\MM[q] = -\MM[\Delta q+k^2q],\\
   (\rot\rot-k^2)\NN[q] = -\NN[\Delta q+k^2q].
\end{gathered}
\end{equation}
Thus, looking for solutions of \eqref{1E1b} amounts to considering $\MM[q]$ and $\NN[q]$ with $q$ solution of the Helmholtz equation $\Delta q + \kappa^2q=0$.

\section{Electric eigenmodes in a product domain}
\label{S3}

In this section, we look for solutions $(k^2,\EE)$ of the electric problem \eqref{1E1} with the gauge constraint $\Div\EE=0$. For this we use the $\MM$, $\NN$ ansatz, we find sufficient conditions on the potentials $q$, construct families of eigenpairs and prove that this system is complete.

\subsection{TE modes}

Let $\EE=\MM[q]$ be a TE mode. By construction $\Div\EE=0$. By \eqref{eq:MNH}, $q$ has to satisfy 
\begin{equation}
\label{eq:Hel}
   \Delta q + \kappa^2q=0.
\end{equation}
It remains to verify the electric boundary conditions $\EE\times\nn=0$ on $\partial\Omega$. Combining \eqref{1E4} and \eqref{eq:Mrot}, we find
\[
\begin{split}
   &\rotv q\times\np = 0  \quad\mbox{on}\quad\partial\omega\times I,
   \\
   &\rotv q=0  \quad\mbox{on}\quad\omega\times\partial I,
\end{split}
\]
which is equivalent to
\begin{equation}
\label{1E4E}
\begin{split}
   &\partial_n q = 0  \quad\mbox{on}\quad\partial\omega\times I,
   \\
   &\grap q=0  \quad\mbox{on}\quad\omega\times\partial I.
\end{split}
\end{equation}
Sufficient conditions for this are Dirichlet conditions on $\omega\times\partial I$ combined with Neumann conditions on $\partial\omega\times I$. This is a tensor product of a Neumann problem on $\omega$ and a Dirichlet problem on $I$. Along the same principle than for pure Dirichlet problem, cf Theorem \ref{th:1}, we find a spectral basis for $q$ in the form
\begin{equation}
\label{}
   q_{jm} = v\neu_j\otimes w\dir_m,\quad k^2=\lambda\neu_j+\mu\dir_m,
   \quad j\ge1,\ m\ge1.
\end{equation}
Here $j=0$ (corresponding to $v\neu_0=1$) is discarded because functions $q$ independent of $\xp$ give $\MM[q]=0$.

Thus we have found the following families of TE modes:

\begin{lemma}
\label{2L1}
For all $j\ge1$, $m\ge1$, the field $\EE\TE_{jm}:=\MM[v\neu_j\otimes w\dir_m]$, i.e.
\begin{equation}
\label{TEl}
   \EE\TE_{jm}(\xp,x_3) = \begin{pmatrix}
         \rotv v\neu_j(\xp) \\ 0
         \end{pmatrix} w\dir_m(x_3),
\end{equation}
is a TE mode for problem \eqref{1E1} associated with the eigenvalue 
$\Lambda\TE_{jm} = \lambda\neu_j+\mu\dir_m$.
\end{lemma}

\subsection{TM modes}

Let $\EE=\NN[q]$ be a TM mode. Again, $\Div\EE=0$, $q$ has to satisfy \eqref{eq:Hel}, and it remains to verify the electric boundary conditions $\EE\times\nn=0$ on $\partial\Omega$: Using \eqref{eq:Mrot}, we find that 
\[
   \EE_\perp = \grap(\partial_3q)\quad\mbox{and}\quad
   E_3 = -\Deltap q
\]
Hence, with \eqref{1E4} 
\[
\begin{split}
   &\grap(\partial_3q)\times\np = 0  \quad\mbox{and}\quad \Deltap q=0
   \quad\mbox{on}\quad\partial\omega\times I,
   \\
   &\grap(\partial_3q)=0  \quad\mbox{on}\quad\omega\times\partial I.
\end{split}
\]
We obtain sufficient conditions through the separation of variable ansatz
\[
   q(x) = v(\xp)\,w(x_3)
\]
with
\begin{equation}
\label{2E7}
   -\Deltap v = \lambda v \ \mbox{ in }\ \omega\quad\mbox{and}\quad
   -\partial^2_3 w = \mu w  \ \mbox{ in }\ I\quad\mbox{with}\quad
   \lambda+\mu=k^2=\Lambda,
\end{equation}
and the boundary conditions become
\[
\left\{ \begin{array}{l@{\qquad}l}
   (\np\times\grap) v(\xp) \:\partial_3w(x_3)= 0  
    &\forall\xp\in\partial\omega,\ \forall x_3\in I,
   \\
   \Deltap v(\xp)\: w(x_3)=0
   &\forall\xp\in\partial\omega,\ \forall x_3\in I,
   \\
   \grap v(\xp) \partial_3w(x_3)=0  &\forall\xp\in\omega,\ \forall x_3\in\partial I,
\end{array}\right.
\]
which yields, with $\partial_d\omega$, $d=1,\ldots,D$, the connected components of $\partial\omega$,
\begin{equation}
\label{2E8}
\left\{ \begin{array}{ll}
   v = {\rm const.} \ \ \mbox{on each} \ \ \partial_d\omega
   &\quad\mbox{or}\quad \partial_3w\equiv0  \ \ \mbox{in} \ \ I, \\[0.3ex]
   \Deltap v=0 \ \ \mbox{on}\ \  \partial\omega
   &\quad\mbox{or}\quad w\equiv0  \ \ \mbox{in} \ \ I,\\[0.3ex]
   \grap v \equiv0  \ \ \mbox{in} \ \ \omega
   &\quad\mbox{or}\quad \partial_3 w=0 \ \ \mbox{on} \ \ \partial I.
\end{array}\right.
\end{equation}
The conditions $\grap v \equiv0$ and $w\equiv0$ have to be discarded since they imply $\EE\equiv0$. Therefore we should have $\partial_3 w=0$ on $\partial I$ and $\Deltap v=0$ on $\partial\omega$. The latter condition implies that $v=0$ on $\partial\omega$ in the case when  $\lambda\neq0$. When $\lambda=0$, the condition $v = {\rm const.}$ on each $\partial_d\omega$ is sufficient. Thus we have shown that \eqref{2E7}-\eqref{2E8} can be summarized as follows:
Either
\begin{equation}
\label{2E9}
\left\{ \begin{array}{ll}
   -\Deltap v = \lambda v  \ \ \mbox{in} \ \ \omega \ 
   &\mbox{and}\quad v=0\ \ \mbox{on} \ \ \partial\omega\\[0.8ex]
   -\partial^2_3 w = \mu w  \ \ \mbox{in} \ \ I 
   &\mbox{and}\quad  \partial_3w=0 \ \ \mbox{on} \ \ \partial I\\
\end{array}\right. 
   \quad\mbox{with}\quad\lambda\neq0,\ \ \lambda+\mu=\Lambda,
\end{equation}  
or
\begin{equation}
\label{2E9b}
\left\{ \begin{array}{ll}
   -\Deltap v = 0  \ \ \mbox{in} \ \ \omega \ 
   &\mbox{and}\quad v={\rm const}\ \ \mbox{on each} \ \ \partial_d\omega\\[0.8ex]
   -\partial^2_3 w = \mu w  \ \ \mbox{in} \ \ I 
   &\mbox{and}\quad  \partial_3w=0 \ \ \mbox{on} \ \ \partial I\\
\end{array}\right. \quad\mbox{with}\quad \mu=\Lambda.
\end{equation}  
Hence we have found the following two families of TM modes. First, the standard one:

\begin{lemma}
\label{2L2}
For all $j\ge1$, $m\ge0$, the field $\EE\TM_{jm} := \NN[v\dir_j\otimes w\neu_m]$, i.e.
\begin{equation}
\label{TMl}
   \EE\TM_{jm}(\xp,x_3) = \begin{pmatrix}
         \grap v\dir_j(\xp) \\ 0
         \end{pmatrix} \partial_3 w\neu_m(x_3) -
         \begin{pmatrix}
         0 \\ \Deltap v\dir_j(\xp)
         \end{pmatrix} w\neu_m(x_3),
\end{equation}
is a TM mode for problem \eqref{1E1} associated with the eigenvalue 
$\Lambda\TM_{jm} = \lambda\dir_j+\mu\neu_m$.
\end{lemma}

The second family appears if $\omega$ has a non trivial topology (i.e.\ if $D\ge2$), and shares the features of TE and TM polarization (vanishing third component of the electric and magnetic fields):

\begin{lemma}
\label{2L2TEM}
There exist $D$ linearly independent harmonic potentials $v\too_d$ that have constant traces on each connected component 
$\partial_d\omega$ of $\partial\omega$. They can be chosen such that $v\too_D$ is constant in $\omega$. If $\partial\omega$ has more than one connected component, then the $v\too_d$, $d=1,\ldots,D-1$, have linearly independent gradients, and they generate the family of TEM modes defined for all $d=1,\ldots,D-1$ and $m\ge1$ as the fields $\EE\TEM_{dm} := \NN[v\too_d\otimes w\neu_m]$ which can also be written as
\begin{equation}
\label{TEM}
   \EE\TEM_{dm}(\xp,x_3) = \begin{pmatrix}
         \grap v\too_d(\xp) \\ 0
         \end{pmatrix}  w\dir_m(x_3),
\end{equation}
and is associated with the eigenvalue $\Lambda\TEM_{dm} = \mu\dir_m$.
\end{lemma}

Note that to obtain \eqref{TEM} we have used that the derivatives $\partial_3 w\neu_k$ for $k\ge1$ are an eigenvector basis for the Dirichlet problem on the interval $I$.

\begin{remark}
\label{2R1}
Let us borrow the following objects from \cite{AmroucheBernardiDaugeGirault98}: Let $\omega^\circ$ be $\omega\setminus\Sigma$, where $\Sigma=\cup_{d=1}^{D-1}\Sigma_d$ is a minimal set of cuts so that $\omega^\circ$ is simply connected. Then we can define the space $\Theta(\omega)$ as:
\[
   \Theta(\omega) = \{\varphi\in H^1(\omega^\circ)|\quad 
   \big[\varphi\big]_{\Sigma_d} = {\rm const}(d),\ \ d=1,\ldots, D-1\}.
\]
For $\varphi\in\Theta(\omega)$, its extended $\rotv$ denoted by $\rotw\varphi$ is defined as its $\rotv$ in $\omega^\circ$, considered as an element of $L^2(\omega)$. 
Then there exist ``conjugate'' potentials $\tilde v\too_d\in\Theta(\omega)$ such that for any $d\le D-1$, there holds
\begin{equation}
\label{2E10}
   \rotw\tilde v\too_d = \grap v\too_d.
\end{equation}
Therefore for all $m\ge1$, the mode $\EE\TEM_{dm}$ is also an extended TE mode. This is why it is called a TEM mode.
\end{remark}

\subsection{Completeness}
The aim of this section is to prove
\begin{lemma}
\label{2L3}
Let $\uu\in\XX\el(\Omega)$ such that $\Div\uu=0$. We assume that for all integers $j\ge1$ and $d\in[1,D-1]$
\[
   \langle \uu,\EE\TE_{jm}\rangle =0\; (\forall m\ge1), \quad
   \langle \uu,\EE\TM_{jm}\rangle =0\; (\forall m\ge0)  \;\;\mbox{and}\;\;
   \langle \uu,\EE\TEM_{dm}\rangle =0 \; (\forall m\ge1).
\]
Here $\langle \cdot,\cdot \rangle$ is the $L^2$ scalar product on $\Omega$. Then $\uu=0$.
\end{lemma}

\begin{proof}
We first draw consequences from the orthogonality properties against the TM modes: We fix $j$ and $m$ and set $v=v\dir_j$, $w=w\neu_m$ and integrate by parts:
\[
   \begin{split}
   0 &= \int_I\int_\omega \up(\xp,x_3)\, \grap v(\xp) \partial_3 w(x_3) -
      u_3(\xp,x_3) \,\Deltap v(\xp) w(x_3) \; \dr \xp \dr x_3
\\
   &= \int_I\int_\omega -\Divp \up(\xp,x_3)\, v(\xp) \partial_3 w(x_3) -
      u_3(\xp,x_3) \,\Deltap v(\xp) w(x_3) \; \dr \xp \dr x_3
\\
   &= \int_I\int_\omega \partial_3u_3(\xp,x_3)\, v(\xp) \partial_3 w(x_3) -
      u_3(\xp,x_3) \,\Deltap v(\xp) w(x_3) \; \dr \xp \dr x_3
\\
   &= \int_I\int_\omega -u_3(\xp,x_3)\, v(\xp) \partial^2_3 w(x_3) -
      u_3(\xp,x_3) \,\Deltap v(\xp) w(x_3) \; \dr \xp \dr x_3.
\\
   \end{split}
\]
Here we have used that $\Div\uu=0$, replacing $\Divp\up$ by $-\partial_3u_3$. Coming back to the properties of $v=v\dir_j$ and $w=w\neu_m$ we find for all $j\ge1$ and $m\ge0$
\[
   \int_I\int_\omega u_3(\xp,x_3)\, (\lambda\dir_j+\mu\neu_m) v\dir_j(\xp)  w\neu_m(x_3)
      \; \dr \xp \dr x_3 = 0.
\] 
Since $\lambda\dir_j+\mu\neu_m$ is never $0$, we deduce that for all $j\ge1$ and $m\ge0$
\[
   \int_I\int_\omega u_3(\xp,x_3)\,  v\dir_j(\xp)  w\neu_m(x_3)
      \; \dr \xp \dr x_3 = 0.
\] 
The set $v\dir_j(\xp)  w\neu_m(x_3)$ being a complete basis in $L^2(\Omega)$, we deduce that $u_3=0$.

Next, we use the orthogonality against the TE modes: for all $j\ge1$ and $m\ge1$ there holds:
\[
   \int_I w\dir_m(x_3) \int_\omega 
   \up(\xp,x_3)\cdot \rotv v\neu_j(\xp) \; \dr \xp \dr x_3 = 0.
\]
Therefore, for all $j\ge1$:
\[
   \int_\omega \up(\xp,x_3)\cdot \rotv v\neu_j(\xp) \; \dr \xp = 0,\quad \mbox{for a.\ e. } x_3\in I.
\]
We deduce that $\rots\up(\cdot,x_3)$ is orthogonal to all $v\neu_j$ for $j\ge1$, which means that $\rots\up(\cdot,x_3)$ is constant with respect to $\xp$. There exists a function $z=z(x_3)$ such that
\[
   \rots\up(\xp,x_3) = z(x_3).
\leqno{(*)}\]
Since $\Div\uu=0$ and $u_3=0$, we have $\Divp\up=0$, which implies that locally $\up$ is a $\rotv$ of a scalar potential and that
\[
   \int_{\partial\omega} \up\cdot\np \;\dr\sigma = 0.
\]
Additionally, the orthogonality relations against the TEM modes yield for all $m\ge1$ and $d\le D-1$
\[
   \int_I w\dir_m(x_3) \int_\omega 
   \up(\xp,x_3)\cdot \grap v\too_d(\xp) \; \dr \xp \dr x_3 = 0.
\]
We deduce that 
\[
   \int_\omega 
   \up(\xp,x_3)\cdot \grap v\too_d(\xp) \; \dr \xp = 0,\quad \mbox{for a.\ e. }  x_3\in I, 
\]
from which we find that (we recall that $\partial_d\omega$ are the connected components of $\partial\omega$)
\[
   \int_{\partial_d\omega} \up\cdot\np \;\dr\sigma = 0, \quad d=1,\ldots,D.
\]
These are the flux conditions that provide the existence of a global scalar potential
 $y\in L^2(I,H^1(\omega))$ such that
\[
   \up(\xp,x_3) = \rotv y(\xp,x_3).
\]
As $\up(\cdot,x_3)$ satisfies the tangential boundary condition on $\partial\omega$ for a.e.\ $x_3\in I$, then $y(\cdot,x_3)$ satisfies in turn the Neumann boundary condition on $\partial\omega$ for a.e.\ $x_3\in I$.
With $(*)$ we find
\[
   -\Deltap y(\xp,x_3) = z(x_3).
\]
Since $y$ satisfies the homogeneous Neumann condition with respect to $\xp$, this implies that  $z(x_3)=0$ for all $x_3$. Finally we have obtained that $\up=0$.
\end{proof}

\subsection{Eigenmodes}
Summarizing, we have proved:
\begin{theorem}
\label{2T1}
Let $\Omega=\omega\times I$. The eigenpairs with zero divergence of the electric Maxwell operator \eqref{1E1} can be organized in the three families:\\
(i) \ $\di\EE\TE_{jm} = \begin{pmatrix}
         \rotv v\neu_j(\xp) \\ 0
         \end{pmatrix} w\dir_m(x_3)$
         \ with \  
         $\Lambda\TE_{jm} = \lambda\neu_j+\mu\dir_m,\ \ j\ge1, \ m\ge1,$\\
(ii) \ $\EE\TM_{jm} = \begin{pmatrix}
         \grap v\dir_j(\xp) \\ 0
         \end{pmatrix} \partial_3 w\neu_m(x_3) -
         \begin{pmatrix}
         0 \\ \Deltap v\dir_j(\xp)
         \end{pmatrix} w\neu_m(x_3)$\\[1ex] \mbox{}\hfill with 
         $\Lambda\TM_{jm} = \lambda\dir_j+\mu\neu_m, \ \ j\ge1, \ m\ge0,$\\[0.5ex]
(iii) and, if $\omega$ is not simply connected (i.e.\ $D\ge2$)\\
$\EE\TEM_{dm} = \begin{pmatrix}
         \grap v\too_d(\xp) \\ 0
         \end{pmatrix}  w\dir_m(x_3)$
         \ with \ 
         $\Lambda\TEM_{dm} = \mu\dir_m,\ \ 1\le d\le D-1, \ m\ge1$.\\[0.5ex]
See Notation \ref{not:1}, Lemmas \ref{2L2} and \ref{2L2TEM} for the notations of the 2D and 1D quantities.
All the associated eigenvalues $\Lambda\TE_{jm}$, $\Lambda\TM_{jm}$ and $\Lambda\TEM_{dm}$ are non-zero.
\end{theorem}

\section{Magnetic eigenmodes in a product domain}
\label{S3M}
Since the magnetic field $\HH$ associated with the electric field $\EE$ is given by 
\[
   \HH=\frac1{ik}\; \rot\EE,\quad\mbox{for}\quad k = \pm\sqrt{\Lambda}
\]
for any non-zero eigenvalue $\Lambda$, we deduce:

\begin{corollary}
\label{2C1}
Under the conditions of Theorem \ref{2T1}, we set $k=\pm\sqrt{\Lambda}$. The associated magnetic fields are given by
\[
   \begin{split}
   &\HH\TE_{jm} = \frac{1}{ik\TE_{jm}} \left\{ \begin{pmatrix}
         \grap v\neu_j(\xp) \\ 0
         \end{pmatrix} \partial_3 w\dir_m(x_3) -
         \begin{pmatrix}
         0 \\ \Deltap v\neu_j(\xp)
         \end{pmatrix} w\dir_m(x_3) \right\}
         \quad j, m\ge1 ,\\
   &\HH\TM_{jm} = -ik\TM_{jm} \begin{pmatrix}
         \rotv v\dir_j(\xp) \\ 0
         \end{pmatrix} w\neu_m(x_3)
         \quad j\ge1, \ m\ge0 ,\\
   &\HH\TEM_{dm} = \frac{i}{k\TEM_{dm}} \begin{pmatrix}
         \rotv v\too_d(\xp) \\ 0
         \end{pmatrix}  \partial_3 w\dir_m(x_3)
         \quad 1\le d\le D-1, \ m\ge1 ,
   \end{split}
\]
and the triples $(k,\EE,\HH)$ are Maxwell eigenmodes.
\end{corollary}

\begin{remark}
\label{2R2}
\iti1 The {\em electric} fields in the pairs $(\EE\TE,\HH\TE)$ are transverse to the axis $x_3$, whilst in the pairs $(\EE\TM,\HH\TM)$ the {\em magnetic} fields are transverse to the axis $x_3$, which justifies the labels of the polarizations.

\iti2 We notice that for all $m\ge1$, $\HH\TEM_{dm}$ can also be written as
\[
   \HH\TEM_{dm} = i \begin{pmatrix}
         \rotv v\too_d(\xp) \\[1ex] 0
         \end{pmatrix}  w\neu_m(x_3).
\]
The expression above also makes sense for $m=0$. The associated eigenvalue is $0$ and the corresponding electric field is $0$. These magnetostatic Maxwell eigenmodes $(0,{\bf0},\HH\TEM_{d\, 0})$ are those produced by the 3D topological non-triviality of $\Omega$. 
\end{remark}

\begin{remark}
\label{2R3}
If $\omega$ contains holes, i.e.\ if TEM modes are present, they often contribute the smallest  positive eigenvalues. Let us make formulas for eigenvalues more explicit: Let $\ell$ be the length of the interval $I$ and let us assume that $\omega$ has \emph{one hole}. Besides the magnetostatic zero eigenvalue, we find
\[
   \Lambda\TE_{jm} = \lambda\neu_j + \Big( \frac{m\pi}{\ell} \Big)^2 \; 
   (\forall j,m\ge1),\quad   
   \Lambda\TM_{jm} = \lambda\dir_j + \Big( \frac{m\pi}{\ell} \Big)^2 \; 
   (\forall j\ge1,m\ge0),\quad   
\]
and
\[
   \Lambda\TEM_{m} =  \Big( \frac{m\pi}{\ell} \Big)^2 \; 
   (\forall m\ge1).
\]
Then the smallest positive eigenvalue is either $\Lambda\TM_{1,0}$ or $\Lambda\TEM_{1}$. If $\omega$ is fixed and $\ell$ large enough, $\Lambda\TEM_{1}$ is smaller than $\Lambda\TM_{1,0}$.
\end{remark}
We summarize the results of sections \ref{S3} and \ref{S3M} in Table \ref{tab:1}.

\begin{table}[h]
\begin{center}
\begin{tabular}{c|ccc}
Polarization & $k^2$ & $\EE$ & $\HH$ \\
\hline
TE & $\lambda\neu_j+(\frac{m\pi}{\ell})^2$ 
   & $\MM[v\neu_j\otimes\sin(\frac{m\pi}{\ell}\cdot) ]$
   & $\frac{1}{ik}\,\NN[v\neu_j\otimes\sin(\frac{m\pi}{\ell}\cdot)]$
   \\
TM & $\lambda\dir_j+(\frac{m\pi}{\ell})^2$ 
   & $\NN[v\dir_j\otimes\cos(\frac{m\pi}{\ell}\cdot) ]$
   & $ik\,\MM[v\dir_j\otimes\cos(\frac{m\pi}{\ell}\cdot) ]$
   \\
TEM & $(\frac{m\pi}{\ell})^2$ 
    & $\NN[v\too_d\otimes\cos(\frac{m\pi}{\ell}\cdot) ]$
    & $ik\,\MM[v\too_d\otimes\cos(\frac{m\pi}{\ell}\cdot) ]$ \\
Magnetostatic
   & $0$
   & ${\bf0}$
   & $\MM[v\too_d\otimes 1]$ 
\end{tabular}
\caption{Synthetic description of Maxwell eigenmodes, using $\MM$ and $\NN$ \eqref{eq:MN}.}
\label{tab:1}
\end{center}
\end{table}

\section{Mixed perfectly conducting or insulating conditions}
\label{S3x}
Consider now the situation where a part $\partial\Omega\cnd$ of the boundary of $\Omega$ represents {\em perfectly conducting walls} whereas another part $\partial\Omega\ins$ represents {\em perfectly insulating} walls, with
\begin{equation}
\label{0Eins}
   \partial\Omega = \partial\Omega\cnd \cup \partial\Omega\ins,\quad
   \partial\Omega\cnd \cap \partial\Omega\ins = \emptyset.
\end{equation}
Boundary conditions are then
\[
   \left\{
   \begin{array}{lll}
   \EE\times\nn = 0 \quad\mbox{and}\quad \HH\cdot\nn=0, 
   & \mbox{on}\quad\partial\Omega \cnd,
   &\mbox{(perfect conductor b.\ c.)}\\
   \EE\cdot\nn = 0 \quad\mbox{and}\quad \HH\times\nn=0, 
   & \mbox{on}\quad\partial\Omega \ins,
   &\mbox{(perfect insulator b.\ c.)}\\
   \end{array}
   \right.
\]

Similar results as above hold for {\em mixed boundary conditions} when the perfectly conducting or insulating parts $\partial\Omega\cnd$ and $\partial\Omega\ins$ are chosen to be either $\partial\omega\times I$ or $\omega\times\partial I$. Let us give two examples.

\begin{example}
Let us consider the case when
\[
   \partial\Omega\cnd = \partial\omega\times I\quad\mbox{and}\quad
   \partial\Omega\ins = \omega\times\partial I.
\]
Then, the essential boundary condition for the electric field $\EE$ on $\omega\times\partial I$ is $\EE_3=0$ and the natural boundary condition is $\rot\EE\times\nn=0$, reducing to $\partial_3\EE_\perp=0$. Thus we find the three families of electric eigenfunctions:
\[
   \begin{split}
   &\EE\TE_{jm} = \begin{pmatrix}
         \rotv v\neu_j(\xp) \\ 0
         \end{pmatrix} w\neu_m(x_3)
         \ \ \mbox{with}\ \  
         j\ge1, \ m\ge0, \\
   &\EE\TM_{jm} = \begin{pmatrix}
         \grap v\dir_j(\xp) \\ 0
         \end{pmatrix} \partial_3 w\dir_m(x_3) -
         \begin{pmatrix}
         0 \\ \Deltap v\dir_j(\xp)
         \end{pmatrix} w\dir_m(x_3),
         \ \ \mbox{with}\ \   j\ge1, \ m\ge1,\\
   &\EE\TEM_{dm} = \begin{pmatrix}
         \grap v\too_d(\xp) \\ 0
         \end{pmatrix}  w\neu_m(x_3)
         \ \ \mbox{with}\ \  
        1\le d\le D-1, \ m\ge0 .
   \end{split}
\]
associated with the eigenvalues $\Lambda\TE_{jm} = \lambda\neu_j+\mu\neu_m$, $\Lambda\TM_{jm} = \lambda\dir_j+\mu\dir_m$, and $\Lambda\TEM_{dm} = \mu\neu_m$.
\end{example}

\begin{example}
 We set $I=(0,\ell)$. Let us consider the case when
\[
   \partial\Omega\cnd = (\partial\omega\times I) \cup (\omega\times \{0\})
   \quad\mbox{and}\quad
   \partial\Omega\ins = \omega\times\{\ell\}.
\]
The axial generators $w_m$ can be described thanks to the eigenvectors $w\mix_m$, $m\ge1$, of the {\em mixed} problem in $\omega$:
\[
   -\partial^2_3w = \mu w,\quad w(0)=0,\quad \partial_3 w(\ell)=0.
\]
We find
\[
   \begin{split}
   &\EE\TE_{jm} = \begin{pmatrix}
         \rotv v\neu_j(\xp) \\ 0
         \end{pmatrix} w\mix_m(x_3)
         \ \ \mbox{with}\ \  
         j\ge1, \ m\ge1, \\
   &\EE\TM_{jm} = \begin{pmatrix}
         \grap v\dir_j(\xp) \\ 0
         \end{pmatrix} \partial^2_3 w\mix_m(x_3) -
         \begin{pmatrix}
         0 \\ \Deltap v\dir_j(\xp)
         \end{pmatrix} \partial_3w\mix_m(x_3),
         \ \ \mbox{with}\ \   j\ge1, \ m\ge1,\\
   &\EE\TEM_{dm} = \begin{pmatrix}
         \grap v\too_d(\xp) \\ 0
         \end{pmatrix}  w\mix_m(x_3)
         \ \ \mbox{with}\ \  
        1\le d\le D-1, \ m\ge1 .
   \end{split}
\]
If $\omega$ contains holes, TEM modes are present and contribute the smallest  positive eigenvalue $\big( \frac{\pi}{2\ell} \big)^2$.
\end{example}

\section{Application 1: Maxwell eigenvalues of cuboids}
\label{S4}

\subsection{Cube}
Let $\Omega$ be the cube $(0,\pi)^3$. We can apply Theorem \ref{2T1} with $\omega=(0,\pi)^2$ and $I=(0,\pi)$. Since $\omega$ is simply connected we have TE and TM modes only. Therefore the normalized Maxwell eigenvalues are
\[
   \lambda\neu_j+\mu\dir_m,\ \ j\ge1, \ m\ge1
   \quad\mbox{and}\quad
   \lambda\dir_j+\mu\neu_m, \ \ j\ge1, \ m\ge0.
\]
We have
\[
   \mu\dir_m = m^2,\ \ m\ge1
   \quad\mbox{and}\quad
   \mu\neu_m = m^2,\ \ m\ge0.
\]
The Dirichlet eigenvalues on $\omega$ are
\[
   k^2_1 + k^2_2,\quad k_1,k_2\ge1.
\]
The non-zero Neumann eigenvalues are
\[
   k^2_1 + k^2_2,\quad k_1,k_2\ge0, \ \ k_1\ \mbox{or}\ k_2\neq0.
\]
Therefore the TE eigenvalues are
\[
   k^2_1 + k^2_2 + k^2_3,\quad k_1,k_2\ge0, \ \ k_1\ \mbox{or}\ k_2\neq0, \ \ k_3\ge1.
\]
The TM eigenvalues are
\[
   k^2_1 + k^2_2 + k^2_3,\quad k_1,k_2\ge1, \ \ k_3\ge0.
\]
Therefore we have once
\[
   k^2_1 + k^2_2 + k^2_3,\ \ k_1,k_2,k_3\ge0 \ \ 
   \mbox{with exactly one index $\nu\in\{1,2,3\}$ such that}\ \ k_\nu=0,
\]
and twice
\[
   k^2_1 + k^2_2 + k^2_3,\quad k_1,k_2,k_3\ge1.
\]
The first eigenvalues are
\[
   2 \ \mbox{ (mult. 3),}\quad 3 \ \mbox{ (mult. 2),}\quad 5 \ \mbox{ (mult. 6),}\quad
   6 \ \mbox{ (mult. 6),}\quad 8 \ \mbox{ (mult. 3),}...
\]
A larger multiplicity of 12 is attained for example for $14=1+4+9$. But 12 is not the maximal multiplicity (e.g.\ the multiplicity of $26=25+1+0=16+9+1$ is 18).

The Dirichlet eigenvectors on $(0,\pi)$ are $\zeta \mapsto\sin k\zeta$, $k\ge1$, and the Neumann eigenvectors are $\cos k\zeta$, $k\ge0$. The components of the electric eigenvectors in the cube are (sums of) products of two $\sin$ terms by one $\cos$ term.

\subsection{Cuboids}
For a \emph{rectangular parallelepiped} 
$$\Omega=(0,\ell_1)\times(0,\ell_2)\times(0,\ell_3),$$ we find the eigenvalues: Once
\begin{multline*}
   \left(\frac{k_1\pi}{\ell_1}\right)^2 + 
   \left(\frac{k_2\pi}{\ell_2}\right)^2 + 
   \left(\frac{k_3\pi}{\ell_3}\right)^2,\\ 
   \forall k_1,k_2,k_3\ge0 \quad
   \mbox{with exactly one index $\nu\in\{1,2,3\}$ such that}\ \ k_\nu=0,
\end{multline*}
and twice
\[
   \left(\frac{k_1\pi}{\ell_1}\right)^2 + 
   \left(\frac{k_2\pi}{\ell_2}\right)^2 + 
   \left(\frac{k_3\pi}{\ell_3}\right)^2,\quad \forall k_1,k_2,k_3\ge1.
\]

\section{Application 2: Maxwell eigenvalues in axisymmetric product domains}
\label{S5}

We assume now, besides the assumption that $\Omega=\omega\times I$, that the domain $\Omega$ is axisymmetric. 
In this case, the separation of variables method can be used once more, giving explicit formulas for the Laplace eigenvectors and eigenfunctions and hence more explicit formulas for the Maxwell eigenmodes. Now $\Omega$ axisymmetric implies that $\omega$ is an axisymmetric domain in dimension 2. Hence $\omega$ is either a disc or an annulus (i.e. a disc with a concentric hole). We investigate both situations. 

\subsection{Axisymmetric domains}
Let $R$ be the external radius of $\omega$ and $r_0$ be its internal radius, with the convention that $r_0=0$ corresponds to the case when $\omega$ is a disc. Let us denote by $\T$ the one-dimensional torus
\[
   \T = \R/(2\pi\Z).
\]
We use \emph{cylindrical coordinates} $(r,\varphi,x_3)\in(r_0,R)\times\T\times I$. Setting 
\[
   \check u(r,\varphi,x_3) = u(x),
\] 
we introduce \emph{cylindrical components} $(u_r,u_\varphi,u_3)$ of the field $\uu=(u_1,u_2,u_3)$,
\[
   u_r = \check u_1 \cos\varphi + \check u_2\sin\varphi \quad\mbox{and}\quad
   u_\varphi = - \check u_1\sin\varphi +  \check u_2\cos\varphi.
\]
In particular, for a scalar function $q$, the radial and angular components of $\grap q$ are $\partial_r q$ and $\frac1r \partial_\varphi q$, and those of $\rotv q$ are $\frac1r \partial_\varphi q$ and $-\partial_r q$. With this, we find the representation in cylindrical coordinates of the ansatz $\MM[q]$ and $\NN[q]$ when $q$ has the tensor form $v\otimes w$:
\begin{equation}
\label{5E1}
   \left\{\begin{array}{l}
   \MM_r[v\otimes w] = \frac1r \partial_\varphi v(r,\varphi)\, w(x_3),\\[1ex]
   \MM_\varphi[v\otimes w] = - \partial_r  v(r,\varphi)\, w(x_3),\\[1ex]
   \MM_3[v\otimes w] = 0,
   \end{array}\right.
\end{equation}
and
\begin{equation}
\label{5E2}
   \left\{\begin{array}{l}
   \NN_r[v\otimes w] = \partial_r v(r,\varphi)\, \partial_3 w(x_3),\\[1ex]
   \NN_\varphi[v\otimes w] = \frac1r \partial_\varphi v(r,\varphi)\, \partial_3 w(x_3),\\[1ex]
   \NN_3[v\otimes w] = -\frac1{r^2} ((r\partial_r)^2  + \partial^2_\varphi) v(r,\varphi) \, w(x_3).
  \end{array}\right.
\end{equation}
To describe the Maxwell eigenmodes in the axisymmetric case, we use Table \ref{tab:1} and make explicit the Dirichlet and Neumann eigenvectors $v\dir$ and $v\neu$ on $\omega$, and also $v\too$ when there is a hole ($r_0>0$). 

It is a classical technique to use the invariance under rotation of the Laplace operator $\Deltap$ for diagonalizing it by Fourier series with respect to $\varphi\in\T$. This leads to the following representations:
\begin{equation}
\label{eq:ax1}
   v\dir = h\dir_{np}(r)\,e^{in\varphi} \quad\mbox{and}\quad
   v\neu = h\neu_{np}(r)\,e^{in\varphi},\quad n\in\Z,\ \ p\ge1.
\end{equation}
where for each $n\in\Z$, the functions $(h\dir_{np})_p$ and $(h\neu_{np})_p$ are bases of eigenfunctions for the operator
\begin{equation}
\label{eq:ax2}
   h \longmapsto \Big(-\partial^2_r -\frac{1}{r} \partial_r  + \frac{n^2}{r^2}\Big) h,\quad r\in(r_0,R)
\end{equation}
with appropriate boundary conditions.

\subsection{The cylinder ($\omega$ is a disc)}
For $h\dir$ the boundary condition at $R$ is $h\dir(R)=0$, for $h\neu$ this is $\partial_r h\neu(R)=0$. At the other end $r=0$ of the interval $(0,R)$, the boundary conditions are driven by integrability properties, cf \cite{BernardiDaugeMaday99}: For $h\dir$ and $h\neu$, they are
\begin{equation}
\label{eq:ax3}
   \partial_rh(0)=0 \quad\mbox{if}\quad n=0, \quad\mbox{and}\quad
    h(0)=0 \quad\mbox{if}\quad n\neq0,\quad 
\end{equation}
As a consequence, both $h\dir$ and $h\neu$ are given by the Bessel functions of the first kind $J_n$ that satisfy \eqref{eq:ax3} and the equation $(-\partial^2_r -\frac{1}{r} \partial_r  + n^2) J_n=J_n$, cf \eqref{eq:ax2}. One finds

\begin{lemma}[\cite{CourantHilbert53}]
\label{6L1}
\iti1 Let $(z_{np})_{p\ge1}$ be the increasing sequence of the positive zeros of $J_n$. Then a spectral sequence for the Dirichlet problem for $-\Deltap$ on $\omega$ is
\begin{equation}
\label{eq:ax4}
   \lambda\dir_{np} = \Big( \frac{z_{np}}{R} \Big)^2\quad\mbox{and}\quad
   v\dir_{np} = J_n\Big( \frac{z_{np}\,r}{R} \Big)\,e^{in\varphi},\quad n\in\Z,\ \ p\ge1
\end{equation}

\iti2 Let $(z'_{np})_{p\ge1}$ be the increasing sequence of the positive zeros of $J'_n$. Then a spectral sequence for the Neumann problem for $-\Deltap$ on $\omega$ is, in addition to the constant eigenfunction,
\begin{equation}
\label{eq:ax5}
   \lambda\neu_{np} = \Big( \frac{z'_{np}}{R} \Big)^2\quad\mbox{and}\quad
   v\neu_{np} = J_n\Big( \frac{z'_{np}\,r}{R} \Big)\,e^{in\varphi},\quad n\in\Z,\ \ p\ge1
\end{equation}
\end{lemma}

We summarize results in Table \ref{tab:ax1}.
\medskip

{
\renewcommand{\arraystretch}{1.5}
\begin{table}[h!]
\begin{center}
\begin{tabular}{c|ccc}
Polarization & $k^2$ & $\EE$ & $\HH$ \\
\hline
TE & $(\frac{z'_{np}}{R})^2+(\frac{m\pi}{\ell})^2$ 
   & $\MM[v\neu_{np}\otimes\sin(\frac{m\pi}{\ell}\cdot) ]$
   & $\frac{1}{ik}\,\NN[v\neu_{np}\otimes\sin(\frac{m\pi}{\ell}\cdot)]$
   \\
TM & $( \frac{z_{np}}{R} )^2+(\frac{m\pi}{\ell})^2$ 
   & $\NN[v\dir_{np}\otimes\cos(\frac{m\pi}{\ell}\cdot) ]$
   & $ik\,\MM[v\dir_{np}\otimes\cos(\frac{m\pi}{\ell}\cdot) ]$
\end{tabular}
\caption{Maxwell eigenmodes in a cylinder of radius $R$ and length $\ell$, using $\MM$ \eqref{5E1} -- $\NN$ \eqref{5E2}, and $v\dir_{np}$ \eqref{eq:ax4} -- $v\neu_{np}$ \eqref{eq:ax5}.}
\label{tab:ax1}
\end{center}
\end{table}
}

We give in Table \ref{tab:ax2} values for the first three zeros $z_{n,j}$ and $z'_{n,j}$ for $n=0,1,2$. We use the relation $J_{\nu-1}-J_{\nu+1} = 2J'_\nu$ to compute $z'_{n,j}$. Since $J_{-1} = -J_1$, we note that there holds
\[
   z'_{0,j} = z_{1,j}, \quad\forall j\ge1.
\]

\begin{table}[h!]
\begin{center}
\begin{tabular}{|@{\quad} c @{\quad}|@{\quad} c @{\quad}|@{\quad} c @{\quad}|@{\quad} c @{\quad}|@{\quad} c @{\quad}|@{\quad} c @{\quad}|}
\hline
     $z_{0,j}$ & $z_{1,j}$ & $z_{2,j}$ & 
     $z'_{0,j}$ & $z'_{1,j}$ & $z'_{2,j}$ \\ \hline 
     2.4048 & 3.8317 & 5.1356 & 3.8317 & 1.8412 & 3.0542  \\
     5.5201 & 7.0156 & 8.4172 & 7.0156 & 5.3314 & 6.7061  \\
     8.6537 & 10.173 & 11.620 & 10.173 & 8.5363 & 9.9695  \\
\hline
\end{tabular}

\vskip 1ex
\caption{The first three zeros of $J_0$, $J_1$, $J_2$, $J'_0$, $J'_1$, $J'_2$.}
\label{tab:ax2}
\end{center}
\end{table}

\subsection{The coaxial cylindrical hole ($\omega$ is an annulus)}
In this case again, there exist explicit formulas for the Laplace eigenvectors and eigenfunctions. This is classical knowledge, see for example \cite{Gottlieb79,Gottlieb85}. 
The boundary conditions on $h\dir$ and $h\neu$ are now the standard ones at $r_0$ and $R$. We have to find the associated eigenpairs of the operator \eqref{eq:ax2} for any $n\in\N$. We find that the radial eigenvectors $h\dir$ and $h\neu$ are linear combinations of the Bessel functions $J_n$ and $Y_n$ of first and second kind:
\[
   h\dir_{np}(r) = \alpha_{np} J_n( k_{np}\,r ) 
   + \beta_{np} Y_n(k_{np}\,r )
\]
with eigenvalues $\lambda\dir_{np}=(k_{np})^2$,
where $k_{np}$ are the positive zeros of the determinant function
\[
   k\,\longmapsto\, J_n(kr_0)\,Y_n(kR) - Y_n(kr_0)\,J_n(kR).
\]
Analogous formulas exist for $h\neu$.

Since $\omega$ has one hole, the number $L$ of the connected components of its boundary is $2$. There exists a non-constant harmonic potential $v\too$ that takes two distinct constant values on the two connected components of $\partial\omega$. This generator $v\too$ can be defined as
\[
   v\too(\xp) = \log r.
\]
In connection with Remark \ref{2R1}, we note that the conjugate potential $\tilde v\too$ is the function $\xp\mapsto\varphi$. In cylindrical components, there holds
\[
   \begin{pmatrix}
   \rotw\tilde v\too \\0
   \end{pmatrix}
   = 
   \begin{pmatrix}
   \grap v\too \\ 0
   \end{pmatrix}
   = 
   \begin{pmatrix}
   \frac1r\\0\\{\:\,0\,\:}
   \end{pmatrix}
   \quad\mbox{and}\quad
   \begin{pmatrix}
   \rotv v\too \\ 0
   \end{pmatrix}
   = 
   -\begin{pmatrix}
   {\:\,0\,\:}\\ \frac1r\\0
   \end{pmatrix}.
\]
We summarize the results concerning TEM modes for $\Omega=\omega\times I$ with the annulus $\omega$:
\begin{corollary}
\label{5C2}
Let $\ell$ be the length of the cylinder $\Omega$ with coaxial hole. 
Its family of TEM modes is axisymmetric and has the form $(\frac{m\pi}{\ell},\EE\TEM,\HH\TEM)$ with\\ 
(a) \ for $m\ge1$,
\begin{equation}
\label{5E15}
   \left\{\begin{array}{l}
   E\TEM_r =  \frac1r \sin(\frac{m\pi}{\ell} x_3),\\[1ex]
   E\TEM_\varphi = 0,\\[1ex]
   E\TEM_3 = 0,
  \end{array}\right.
   \quad\mbox{and}\quad
   \left\{\begin{array}{l}
   H\TEM_r = 0,\\[1ex]
   H\TEM_\varphi = -i\frac{m\pi}{\ell}\frac 1r \cos(\frac{m\pi x_3}{\ell}),\\[1ex]
   H\TEM_3 = 0
  \end{array}\right.   
\end{equation}
(b) \  for $m=0$, \ $\EE={\bf0}$ \ and \ $\HH=(0\; 1\; 0)^\top$.
\end{corollary}

\begin{remark}
\label{5R1}
As $r_0$ tends to $0$, the Dirichlet and Neumann eigenmodes of the annulus tend to the  Dirichlet and Neumann eigenvalues of the disc of same radius $R$. Hence the TE and TM modes of the cylinder with hole tend to the TE and TM modes of the cylinder without hole. In contrast, the TEM modes do not depend on $r_0$ as long as $r_0\neq0$, but disappear at the limit when $r_0=0$. This fact has a practical importance when thin conductor wires are present.
\end{remark}

\section{Maxwell eigenmodes in a ball}
\label{sec:sph}
For the sake of comparison, we revisit known results about Maxwell eigenmodes in a ball, see \cite[Chapter 10]{HansonYakovlev}. Let $\Omega\subset\R^3$ be the ball of center $0$ and radius $R$. Here we use spherical coordinates $(\theta,\varphi,\rho)\in[0,\pi]\times\T\times[0,R]$, associate with the orthonormal basis
\[
   (\hat\theta,\hat\varphi,\hat\rho).
\]
Formulas for Maxwell eigenmodes are based on Debye potentials. This is the $\MM$, $\NN$ ansatz, in a form slightly different from \eqref{eq:MN}: The piloting vector is replaced by the unit field
\[
   \hat \bx = \frac{\bx}{|\bx|}\quad\mbox{i.e.}\quad\hat\bx=\hat\rho.
\] 
The $\MM$, $\NN$ ansatz takes the form
\begin{equation}
\label{eq:MNS}
   \MM[q] = \rot(q\,\hat\bx) \quad\mbox{and}\quad \NN[q] = \rot\MM[q] = \rot\rot(q\,\hat\bx).
\end{equation}
Using for instance identities, cf \cite[\S6.2]{CoDa2000ARMA},
\[
   \rot(p\bx) = \grad p \times\bx \quad\mbox{and}\quad
   \rot(\ba\times\bx) = (\rho\partial_\rho+2)\ba - \bx\Div\ba
\]
we find the following formulas where we express vectors in spherical components on the basis $(\hat\theta,\hat\varphi,\hat\rho)$ 
\begin{equation}
\label{eq:sph1}
   \MM[q] = \grad\big(\frac{q}{\rho}\big) \times\bx = \grad q \times\hat\bx =
   \begin{pmatrix}
   \frac{1}{\rho}\partial_\theta q \\ - \frac{1}{\rho\sin\theta}\partial_\varphi q\\0
\end{pmatrix}
\end{equation}
and
\begin{equation}
\label{eq:sph2}
   \NN[q] = \rot\MM[q] = \grad(\partial_\rho q) - \hat\bx \: \rho\,\Delta\big(\frac{q}{\rho}\big)\,.
\end{equation}
Therefore
\begin{equation}
\label{eq:sph3}
   \rot\NN[q] = \rot\rot\MM[q] = -\MM\Big[\rho\,\Delta\big(\frac{q}{\rho}\big)\Big]\,.
\end{equation}
Introduce the operator
\[
   \gL : q\longmapsto \gL q = -\rho\,\Delta\big(\frac{q}{\rho}\big)\:.
\]
Then the equations $\rot\rot\uu-k^2\uu=0$ for $\MM[q]$ and $\NN[q]$ are equivalent to
\[
   \MM[(\gL-k^2)q] = 0 \quad\mbox{and}\quad \NN[(\gL-k^2)q] = 0.
\]
Thus we are interested in scalar solutions $q$ of the equation
\begin{equation}
\label{eq:sph4}
   \gL q  = k^2q \quad\mbox{in}\quad [0,\pi]\times\T\times[0,R].
\end{equation}
We note that 
\begin{equation}
\label{eq:sph5}
   \gL = -\partial^2_\rho  -  \frac{1}{\rho^2}\, \Delta_{\SS^2}
\end{equation}
with the Laplace-Beltrami operator $\Delta_{\SS^2}$ on the unit sphere $\SS^2$
\[
   \Delta_{\SS^2} =    \frac{1}{\sin \theta } \,\partial_\theta \sin \theta \partial_\theta   +
   \frac{1}{\sin^2\theta } \,\partial^{2}_\varphi \,.
\]
The equation \eqref{eq:sph4} is satisfied by all functions in tensor form
\[
   q(\theta,\varphi,\rho) = Y_n^m(\theta,\varphi)\,h(k\rho)\,,
\]
where $Y_n^m$ are the spherical harmonics and $h$ is a linear combination of the Riccati-Bessel functions $\psi_n$ and $\chi_n$ (sometimes written as $S_n$ and $C_n$). Following Debye's notation, we use the definition
\[
   \psi_n(x) = x j_n(x) = \sqrt{\frac{\pi x}{2}} J_{n+\frac{1}{2}}(x) \quad\mbox{and}\quad
   \chi_n(x) = -x y_n(x) = -\sqrt{\frac{\pi x}{2}} Y_{n+\frac{1}{2}}(x),
\]
where $J_\nu$, $Y_\nu$ are the Bessel functions and $j_n$, $y_n$ the spherical Bessel functions of first, second kind, respectively. Because of integrability conditions in $0$, $\chi_n$ has to be discarded. It remains to look for potentials $q$ of type $Y_n^m\otimes \psi_n(k\,\cdot)$ so that either $\MM[q]$ or $\NN[q]$ satisfy the electric boundary condition on the boundary of the ball, i.e.
\[
   \EE\times\nn=0\quad\mbox{if}\quad \rho=R.
\]
Using formulas \eqref{eq:sph1} and \eqref{eq:sph2}, we find that this boundary condition is satisfied by $\MM[Y_n^m\otimes \psi_n(k\,\cdot)]$ if $\psi_n(kR)=0$ and by $\NN[Y_n^m\otimes \psi_n(k\,\cdot\,)]$ if $\psi'_n(kR)=0$. The remarkable fact is that the related potentials are then eigenvectors of the operator $\gL$ with the eigenvalue $k^2$ for Dirichlet or Neumann conditions. Note that the operator $\gL$ is associated with the coercive bilinear form
\[
   a(q,\tilde{q}) = \int_0^R\bigg[\int_{\SS^2} \Big(\partial_\rho q \: \partial_\rho \tilde{q} 
   + \frac{1}{\rho^2} \,\partial_\theta q \:\partial_\theta \tilde{q}
   + \frac{1}{\rho^2\sin^2\theta}\, \partial_\varphi q \:\partial_\varphi \tilde{q}\Big)
   \:\sin\theta\, \rd\theta\, \rd\varphi \bigg] \rd\rho
\]
on the space
\[
   V = \{q\in L^2(\SS^2\times(0,R)),\quad \partial_\rho q,\ \frac1\rho \grap q \in L^2(\SS^2\times(0,R))\}.
\]
Completed with either Dirichlet or Neumann boundary conditions on $\rho=R$, $\gL$ is selfadjoint.
We have obtained:

\begin{theorem}
\label{th:ball}
\iti1 The Dirichlet eigenpairs of $\gL$ have the form $(k_{np}^2,q\dir_{nmp})$, $n\ge0$, $|m|\le n$, $p\ge1$ with $(k_{np})_{p\ge1}$ the enumeration of the positive zeros of the function $k\mapsto \psi_n(kR)$ and $q_{nmp}= Y_n^m\otimes \psi_n(k_{np}\,\cdot\,)\,$. All triples $(k_{np},\MM[q\dir_{nmp}],\frac{1}{ik}\NN[q\dir_{nmp}])$ are Maxwell eigenmodes on the ball of radius $R$.
\smallskip

\iti2 The non-constant Neumann eigenpairs of $\gL$ have the form $((k'_{np})^2,q\neu_{nmp})$, $n\ge0$, $|m|\le n$, $p\ge1$ with $(k'_{np})_{p\ge1}$ the enumeration of the positive zeros of the function $k\mapsto \psi'_n(kR)$ and $q_{nmp}= Y_n^m\otimes \psi_n(k'_{np}\,\cdot\,)\,$. All triples $(k'_{np},\NN[q\neu_{nmp}],ik\MM[q\neu_{nmp}])$ are Maxwell eigenmodes on the ball of radius $R$.
\end{theorem}

\begin{remark}
In the literature, the $\MM$ ansatz is frequently written in a slightly different way which we distinguish with an asterisk: 
\[
   \MM^\star[q^\star] = \rot(q^\star\bx)
\]
instead of $\MM[q]=\rot(q\hat\bx)$. As usual, $\NN^\star=\rot\MM^\star$. The outcome for the Maxwell eigenmodes is the same of course. Nevertheless, the interpretation of the potentials is different. We have
\[
   q^\star = \frac{q}{\rho}.
\]
\begin{enumerate}
\item Concerning Dirichlet modes, the functions $q^\star_{nmp}$ defined as $q\dir_{nmp}/\rho$ are the eigenfunctions of the Dirichlet problem for the standard positive Laplace operator $-\Delta$ on the ball. In other words, the eigenvalues $k_{np}^2$ are also the standard Laplace eigenvalues.

\item But, when Neumann modes are concerned, the functions $q^\star_{nmp}$ defined as $q\neu_{nmp}/\rho$ are not Neumann eigenfunctions for $-\Delta$.
\end{enumerate}

\end{remark}

\begin{remark} 
The tensor product potentials $Y_n^m\otimes h(k\,\cdot\,)$ with $h$ being any of the Riccati-Bessel functions have been used  more than a century ago to describe scattering of plane waves by a dielectric sphere (Mie series). Scattering resonances (with negative imaginary part) have also been investigated at that time. More recently, whispering gallery modes have been analytically calculated by a similar method \cite{BalacFeron2014}. All these problems are transmission problems between the ball and its exterior. Inside the ball $h$ has the form $\psi_n(\mm k\,\cdot\,)$ where $\mm$ is the refractive (or optical) index of the ball. Outside the ball, $h$ is either $\zeta^{(1)}_n(k  \,\cdot\,)$ for scattering, or $\chi_n(k  \,\cdot\,)$ for whispering gallery modes.
\end{remark}

We end this section by a completeness result that can be seen as a consequence of Theorem \ref{th:ball}.
\begin{corollary}
The union of the two families
\[
   \big(k_{np},\MM[q\dir_{nmp}],\frac{1}{ik}\NN[q\dir_{nmp}]\big)_{nmp}
   \quad\mbox{and}\quad
   \big(k'_{np},\NN[q\neu_{nmp}],ik\MM[q\neu_{nmp}]\big)_{nmp}
\]
described in Theorem \ref{th:ball} form a complete set of Maxwell eigenmodes.
\end{corollary}

\begin{proof}
Let $\uu\in\XX\el(\Omega)$ such that $\Div\uu=0$. We assume that $\uu$ is orthogonal to all electric eigenvectors $\MM[q\dir_{nmp}]$ and $\NN[q\neu_{nmp}]$. We prove that $\uu=0$ by contradiction. Assuming that $\uu\neq0$ and relying on the fact that the Maxwell problem possesses an orthonormal basis of eigenfunctions, we may suppose that $\uu$ is an eigenvector itself, associated with an eigenvalue $k^2$. Since the ball $\Omega$ is topologically trivial, the condition $\Div\uu=0$ implies that $k\neq0$, whence $\uu=\frac{1}{k^2}\rot\rot\uu$. The orthogonality of $\uu$ against all  eigenvectors $\MM[q\dir_{nmp}]$ implies through integration by parts that $\rot\uu$ is orthogonal to all $q\dir_{nmp}\,\hat\bx$, hence $\rot\uu$ has a zero radial component. In a similar way, the orthogonality of against all eigenvectors $\NN[q\neu_{nmp}]$ implies that $\rot\rot\uu$, hence $\uu$, has a zero radial component. Finally, the implication
\[
   \uu\cdot\hat\bx=0,\quad \rot\uu\cdot\hat\bx=0,\quad \mbox{and}\quad \Div\uu=0 
   \quad\Longrightarrow\quad
   \uu=0
\]
can be found in \cite{Schulenberger78} and leads to a contradiction, which proves the completeness.
\end{proof}

\section{Extension to nonconstant electric permittivity}
\label{S7}

Let us consider the original Maxwell system \eqref{0E1}. We still assume that the magnetic permeability $\mu$ is equal to $\mu_0$ in the whole domain $\Omega$. But we allow now that the electric permittivity $\varepsilon$ may vary in $\Omega$. We set
\[
   \varepsilon = \varepsilon\rel\varepsilon_0,\quad \varepsilon\rel\ge1.
\]
We consider domains $\Omega$ in the product form $\omega\times I$. We assume that
\begin{equation}
\label{7E0}
   \varepsilon\rel(x) = \varepsilon\rel(\xp),\quad \varepsilon\rel\in L^\infty(\omega),
\end{equation}
like in wave guides or optic fibers. The Maxwell system takes now the form \eqref{7E02} instead of \eqref{0E2}. Then the classification of eigenvectors into TE, TM and TEM does not hold any more (at least not in the form given by Theorem \ref{2T1} and Corollary \ref{2C1}).
Nevertheless, the splitting of the spectrum according to frequencies with respect to the axial variable $x_3$ remains possible, as well as a tensor product form. We are going to investigate the magnetic field $\HH$, taking advantage of its local regularity even if $\varepsilon\rel$ is not continuous.
The magnetic variational formulation becomes, instead of \eqref{1E1m}:\\[1ex] 
{\it Find the eigenpairs $(\Lambda=\kappa^2,\uu)$ with $\uu\ne0$ in $\XX\ma(\Omega)$ and $\Div\uu=0$ such that}
\begin{equation}
\label{7E1m}
    \int_\Omega \frac1{\varepsilon\rel}\, \rot\uu\,\rot\uu' + s\Div\uu\Div\vv\;\dr\xx = 
    \Lambda \int_\Omega \uu\cdot\uu'\;\dr\xx, \quad \forall\uu'\in\XX\ma(\Omega),
\end{equation}
Here $s$ is nonnegative. The choice $s>0$ corresponds to an elliptic regularization of the system.
To simplify notations, let us assume that
\begin{equation}
\label{7E01}
   \framebox{$I = (0,\pi)$}
\end{equation} 
In the constant material case, considering the Maxwell eigenmodes from the magnetic point of view, we note that the magnetic part of eigenmodes given in Corollary \ref{2C1} have the following form
\begin{equation}
\label{7E1}
   \!\!\HH\TE = \begin{pmatrix}
         k\,\grap v(\xp) \, \cos(m x_3)\\ 
         - \Deltap v(\xp)\, \sin(m x_3)
         \end{pmatrix} \quad\mbox{and}\quad
  \HH\TM = \begin{pmatrix}
         \rotv v(\xp)  \, \cos(m x_3) \\ 0
         \end{pmatrix} 
\end{equation}
We are going to prove that we still have a similar structure with respect to the axial variable $x_3$.

\begin{theorem}
\label{7T1}
With the assumptions \eqref{7E0} and \eqref{7E01}, the magnetic eigenmodes solution of \eqref{7E1m} can be organized in a sequence of independent families $\gH_m$ with index $m\in\N$ in which each eigenvector has the tensor product form 
\begin{equation}
\label{eq:7H}
   \HH = \begin{pmatrix}
    \vp(\xp) \, \cos(m x_3)\\ 
     v_{3}(\xp)\, \sin(m x_3)
    \end{pmatrix}.
\end{equation}
For any $m\in\N$, let $\Lambda^m_{j}$ and $\vvkj := (\vpkj,v^m_{3,j})$ be the eigenpairs of the problem:\\[1ex] 
{\it Find $\Lambda\in\R$, $\vv=(\vp,v_3)\ne0$ in $\XX\ma(\omega)\times H^1(\omega)$ with $\Divp\vp + m v_3=0$ such that}
\begin{multline}
\label{7E5}
    \int_\omega \frac1{\varepsilon\rel}\, \Big\{ \rots\vp\,\rots\vp' +
    \big(\grap v_3 + m \vp\big) \cdot  \big(\grap v'_3 + m \vp'\big)
    \Big\}\,\dr\xx \\[-1ex]
    = \Lambda \int_\omega \vv\cdot\vv'\;\dr\xx, \quad 
    \forall\vv'\in\XX\ma(\omega)\times H^1(\omega).
\end{multline}
Denote by $\HH^m_{j}$ the vector of form \eqref{eq:7H} with $\vv=\vvkj$.
Then the eigenpairs $\big(\Lambda^m_{j},\HH^m_{j}\big)_{j\ge1}$ span the family $\gH_m$.
\end{theorem}

\begin{proof}
Solutions of \eqref{7E1m} satisfy on $\omega\times\{0\}$ the essential boundary condition $u_3=0$, and the natural boundary condition $\frac1{\varepsilon\rel} \rot\uu\times\ee_3=0$. Since $u_3=0$ on $\omega\times\{0\}$, $\partial_1u_3$ and $\partial_2u_3$ are also $0$ on $\omega\times\{0\}$, and the natural boundary condition implies that $\partial_3u_1=\partial_3u_2=0$ on $\omega\times\{0\}$. Therefore, defining the extension
\[
   \widetilde \uu_\perp(\xp,-x_3) = \up(\xp,x_3) \quad\mbox{and}\quad
   \tilde u_3(\xp,-x_3) = -u_3(\xp,x_3), \ \ \forall  x_3\in(0,\pi)
\]
we obtain an element $\widetilde\uu\in\XX\ma(\omega\times(-\pi,\pi))$ which satisfies $\Div\widetilde\uu=0$ and is solution of \eqref{7E1m} on the extended domain $\omega\times(-\pi,\pi)$. Moreover, $\uu(\xp,-\pi)=\uu(\xp,\pi)$ and $\partial_3\uu(\xp,-\pi)=\partial_3\uu(\xp,\pi)$ for all $\xp\in\omega$. We deduce that $\widetilde\uu$ is solution of \eqref{7E1m} on the domain $\XX\ma(\omega\times\T)$ where $\T=\R/2\pi\Z$. Since the coefficient $\varepsilon\rel$ does not depend on $x_3$, the underlying Maxwell operator commutes with $\partial_3$. Therefore the spectrum of problem \eqref{7E1m} can be decomposed according to the eigenvectors of $\partial_3$ on $\T$, which are the functions $x_3\mapsto e^{imx_3}$, $m\in\Z$.

For any positive integer $m$, we notice that if $\big(\vp(\xp), v_3(\xp)\big) e^{imx_3}$ is solution of \eqref{7E1m} on the domain $\XX\ma(\omega\times\T)$, then
$\big(\vp(\xp), -v_3(\xp)\big) e^{-imx_3}$ is also solution of the same problem. Therefore, their sum is also solution of the same problem. Moreover this sum has the form \eqref{eq:7H} and satisfies the boundary conditions (perfectly conducting walls)\footnote{%
Considering the \emph{difference} instead the sum, we would find the perfectly insulating boundary conditions on $\omega\times\partial I$.
} 
of the space $\XX\ma(\Omega)$. Conversely this sum is, up to a multiplicative constant, the only linear combination of  $\big(\vp(\xp), v_3(\xp)\big) e^{imx_3}$ and $\big(\vp(\xp), -v_3(\xp)\big) e^{-imx_3}$ which satisfies the boundary conditions of the space $\XX\ma(\Omega)$.

Calculating 
\[
   \int_\Omega \frac1{\varepsilon\rel}\, \rot\uu\,\rot\uu'\;\dr\xx
\]
for
\[
   \uu = \begin{pmatrix}
         \vv(\xp) \, \cos(m x_3)\\ 
         v_3(\xp)\, \sin(m x_3)
         \end{pmatrix}\quad\mbox{and}\quad
   \uu' = \begin{pmatrix}
         \vv'(\xp) \, \cos(m x_3)\\ 
         v'_3(\xp)\, \sin(m x_3)
         \end{pmatrix},
\]
we find
\[
    \int_\omega \frac1{\varepsilon\rel}\, \Big\{ \rots\vp\,\rots\vp' +
    \big(\rotv v_3 + m \vp\times \ee_3\big) \cdot  \big(\rotv v'_3 + m \vp'\times \ee_3\big)
    \Big\}\,\dr\xx 
\]
which coincides with the bilinear form in problem \eqref{7E5}.
\end{proof}

\begin{remark}
\label{7R1}
The bilinear form  of problem \eqref{7E5} can be regularized by
\begin{equation*}
    \int_\omega \frac1{\varepsilon\rel}\, \Big\{
    \big(\Divp\vp + mv_3 \big)   \big(\Divp\vp' + mv_3'\big)
    \Big\}\,\dr\xx.
\end{equation*}
We can check that if $\varepsilon\rel$ is constant, the resulting bilinear form is equal to
\begin{multline*}
   \frac1{\varepsilon\rel} \int_\omega
   \rots\vp\,\rots\vp' +
   \grap v_3 \cdot \grap v'_3  \\[-1ex]
   +  \Divp\vp\, \Divp\vp' + m^2 (\vp \cdot \vp'  +  v_3v'_3) \;\dr\xx.
\end{multline*}
\end{remark}

\begin{remark}
\label{7R0}
For $m=0$, problem \eqref{7E5} reduces to two uncoupled problems: The magnetic 2D Maxwell eigenvalue problem in $\omega$ for $\vp$ and the Neumann eigenvalue problem for $-\Deltap$ in $\omega$ for $v_3$. This last problem does not  yield any non-trivial solution of \eqref{7E5} since for $m=0$, the third component in the Ansatz \eqref{eq:7H} is zero. Moreover, we can show that the solutions of the magnetic 2D Maxwell eigenvalue problem in $\omega$ are the pairs $(\rotv v\dir_j,\lambda\dir_j)$, $j\ge1$, with the eigenpairs $(v\dir_j,\lambda\dir_j)$ of the problem
\begin{equation}
\label{7Es}
   -\Deltap v = \lambda \,\varepsilon v \quad\mbox{in}\quad\omega,\quad v\in H^1_0(\omega).
\end{equation} 
Thus we have found for $m=0$ the family of TM modes:
\[
   \HH\TM_{j} = \begin{pmatrix}
         \rotv v\dir_j(\xp) \\ 0
         \end{pmatrix}
         \qquad j\ge1.
\] 
\end{remark}

\appendix
\section{Normalizing Maxwell equations}
Let $\varepsilon$ and $\mu$ are the electric permittivity and the magnetic permeability of the material inside $\Omega$. We assume that the boundary of $\Omega$ represents {\em perfectly conducting} or {\em perfectly insulating} walls:
\begin{equation}
\label{0Ea}
   \partial\Omega = \partial\Omega\cnd \cup \partial\Omega\ins,\quad
   \partial\Omega\cnd \cap \partial\Omega\ins = \emptyset,
\end{equation}
where $\partial\cnd\Omega$ is the perfectly conducting part and $\partial\ins\Omega$ the perfectly insulating part.
 
The cavity resonator problem is to find the frequencies $\varpi\in\R_+$ and the non-zero electromagnetic fields $(\hat\EE,\hat\HH)\in L^2(\Omega)^6$ such that
\begin{subequations}
\begin{equation}
\label{0E0a}
   \left\{
   \begin{array}{lll}
   \rot \hat\EE - i \varpi\mu \hat\HH = 0 & \mbox{in}\quad\Omega,\quad 
   &\mbox{(Faraday law)}\\
   \rot \hat\HH + i \varpi\varepsilon \hat\EE = 0 & \mbox{in}\quad\Omega, 
   &\mbox{(Amp\`ere law)}\\
   \Div\varepsilon \hat\EE=0 \quad\mbox{and}\quad \Div\mu \hat\HH=0\quad
   & \mbox{in}\quad\Omega,
   &\mbox{(gauge conditions)}.
   \end{array}
   \right.
\end{equation}
with boundary conditions
\begin{equation}
\label{0E0b}
   \left\{
   \begin{array}{lll}
   \hat\EE\times\nn = 0 \quad\mbox{and}\quad \hat\HH\cdot\nn=0, 
   & \mbox{on}\quad\partial\Omega \cnd,
   &\mbox{(perfect conductor b.\ c.)}\\
   \hat\EE\cdot\nn = 0 \quad\mbox{and}\quad \hat\HH\times\nn=0, 
   & \mbox{on}\quad\partial\Omega \ins,
   &\mbox{(perfect insulator b.\ c.)}\\
   \end{array}
   \right.
\end{equation}
\end{subequations}

In this paper, we consider the non-magnetic case, i.e. when $\mu\equiv\mu_0$ in $\Omega$. We can set
\begin{equation}
\label{eq:A1}
   \varepsilon = \mm^2\varepsilon_0 = \varepsilon\rel\,\varepsilon_0
\end{equation} 
where $\mm$ is the refractive index of the material and $\varepsilon\rel$ the relative permittivity. Then \eqref{0E0a} reduces to 
\begin{equation}
\label{0E1}
   \left\{
   \begin{array}{lll}
   \rot \hat\EE - i \varpi\mu_0 \hat\HH = 0 & \mbox{in}\quad\Omega,\quad 
   &\mbox{(Faraday law)}\\
   \rot \hat\HH + i \varpi\varepsilon\rel\,\varepsilon_0 \hat\EE = 0 & \mbox{in}\quad\Omega, 
   &\mbox{(Amp\`ere law)}\\
   \Div\varepsilon\rel\,\varepsilon_0 \hat\EE=0 \quad\mbox{and}\quad \Div\mu_0 \hat\HH=0\quad
   & \mbox{in}\quad\Omega,
   &\mbox{(gauge conditions)}.
   \end{array}
   \right.
\end{equation}
With the normalization
\begin{equation}
\label{0E3}
   k = \varpi\sqrt{\varepsilon_0\mu_0}\quad \mbox{(wave number)},
   \quad  \EE=\sqrt{\varepsilon_0}\, \hat\EE
   \quad\mbox{and}\quad
  \HH=\sqrt{\mu_0}\, \hat\HH.
\end{equation} 
system \eqref{0E1} is transformed into
\begin{equation}
\label{7E02}
   \left\{
   \begin{array}{ll}
   \rot\EE - i k\HH = 0 & \mbox{in}\quad\Omega,\\
   \rot\HH + i k\varepsilon\rel\EE = 0 & \mbox{in}\quad\Omega,\\
   \Div\varepsilon\rel\EE=0 \quad\mbox{and}\quad \Div\HH=0\quad & \mbox{in}\quad\Omega.
   \end{array}
   \right.
\end{equation}


\begin{thebibliography}{10}

\bibitem{AmroucheBernardiDaugeGirault98}
Cherif Amrouche, Christine Bernardi, Monique Dauge and Vivette Girault, Vector
  Potentials in Three-Dimensional Non\-smooth Domains, \emph{Math. Meth. Appl.
  Sci.} {21} (1998), 823--864.

\bibitem{BalacFeron2014}
St{\'e}phane Balac and Patrice F{\'e}ron, \emph{{Whispering gallery modes
  volume computation in optical micro-spheres}}, {FOTON, UMR CNRS 6082},
  Research report, December 2014.

\bibitem{BernardiDaugeMaday99}
Christine Bernardi, Monique Dauge and Yvon Maday, \emph{Spectral methods for
  axisymmetric domains}, Series in Applied Mathematics (Paris)~3,
  Gauthier-Villars, \'Editions Scientifiques et M\'edicales Elsevier, Paris,
  1999.

\bibitem{CoDa1999M2AS}
Martin Costabel and Monique Dauge, Maxwell and {L}am\'e eigenvalues on
  polyhedra, \emph{Math. Methods Appl. Sci.} {22} (1999), 243--258.

\bibitem{CoDa2000ARMA}
\bysame, Singularities of electromagnetic fields in polyhedral domains,
  \emph{Arch. Ration. Mech. Anal.} {151} (2000), 221--276.

\bibitem{CourantHilbert53}
Richard Courant and David Hilbert, \emph{Methods of mathematical physics.
  {V}ol. {I}}, Interscience Publishers, Inc., New York, N.Y., 1953.

\bibitem{Gottlieb79}
Hans P.~W. Gottlieb, Harmonic properties of the annular membrane, \emph{J.
  Acoust. Soc. Amer.} {66} (1979), 647--650.

\bibitem{Gottlieb85}
\bysame, Eigenvalues of the {L}aplacian with {N}eumann boundary conditions,
  \emph{J. Austral. Math. Soc. Ser. B} {26} (1985), 293--309.

\bibitem{HansonYakovlev}
George~W. Hanson and Alexander~B. Yakovlev, \emph{Operator Theory for
  Electromagnetics: An Introduction}, Springer-Verlag, New York, 2002.

\bibitem{Jackson1998}
John~D. Jackson, \emph{Classical Electrodynamics, Third edition}, Wiley, New
  York, 1998.

\bibitem{Schulenberger78}
John~R. Schulenberger, The {D}ebye potential: a scalar factorization for
  {M}axwell's equations, \emph{J. Math. Anal. Appl.} {63} (1978), 502--520.

\end{thebibliography}

\providecommand{\bysame}{\leavevmode\hbox to3em{\hrulefill}\thinspace}
\providecommand{\MR}{\relax\ifhmode\unskip\space\fi MR }
\providecommand{\MRhref}[2]{%
  \href{http://www.ams.org/mathscinet-getitem?mr=#1}{#2}
}
\providecommand{\href}[2]{#2}

\end{document}